\documentclass[aps,pra,reprint,superscriptaddress,amsmath,amssymb,longbibliography]{revtex4-2}

\usepackage{graphicx}
\usepackage{subfigure}
\usepackage{amsthm}
\usepackage{color}
\usepackage{bm}

\newtheorem{Thm}{Theorem}

\newtheorem{Prop}[Thm]{Proposition}
\newtheorem{Cor}[Thm]{Corollary}
\theoremstyle{definition}
\newtheorem{Def}[Thm]{Definition}
\newtheorem{Rem}[Thm]{Remark}
\newtheorem{Ex}[Thm]{Examples}

\newcommand{\ket}[1]{\left|{#1}\right\rangle}
\newcommand{\bra}[1]{\left\langle{#1}\right|}

\newcommand{\ketbra}[2]{\left|{#1}\right\rangle\left\langle{#2}\right|}

\newcommand{\abs}[1]{\left|{#1}\right|}
\DeclareMathOperator{\tr}{Tr}

\begin{document}
\title{Quantum Resources and Performance \\in the Initialization-Free Bernstein-Vazirani Algorithm}
\author{Haesol Han}
    \email{hshan2020@khu.ac.kr}
    \affiliation{
    Department of Mathematics and Research Institute for Basic Sciences,
    Kyung Hee University, Seoul 02447, Korea}
    \affiliation{
    Institute of Fundamental Technological Research, Polish Academy of Sciences, Pawi\'{n}skiego 5B, 02-106 Warsaw, Poland}

\author{Alexander Streltsov}
    \email{streltsov.physics@gmail.com}
    \affiliation{
    Institute of Fundamental Technological Research, Polish Academy of Sciences, Pawi\'{n}skiego 5B, 02-106 Warsaw, Poland}

\author{Soojoon Lee}
    \email{level@khu.ac.kr}
    \affiliation{
    Department of Mathematics and Research Institute for Basic Sciences,
    Kyung Hee University, Seoul 02447, Korea}
    \affiliation{
    School of Computational Sciences, Korea Institute for Advanced Study, Seoul 02455, Korea}

\date{\today} 

\begin{abstract}
Naseri {\it et al.}~[Phys. Rev. A {\bf 106}, 062429 (2022)] studied which quantum resources in initial states are essential for the probabilistic Bernstein-Vazirani (BV) algorithm, defining its performance as the optimal average success probability over all measurements.
In this work, we consider a variant of BV algorithm, which is called the initialization-free (IF) BV algorithm, in which an arbitrary ancilla state as the oracle register is allowed, to improve the performance. 
We derive an explicit formula for the performance of the probabilistic IF-BV algorithm and obtain a necessary and sufficient condition for an initial state to achieve maximal performance. 
We further prove that, under a suitable ordering assumption on the coefficients of the initial state, the probabilistic IF-BV algorithm outperforms the standard probabilistic BV algorithm.
\end{abstract}

\maketitle

\section{Introduction} 
\label{sec:intro}
Quantum algorithms are often studied not only from the viewpoint of query or time complexity, but also through the quantum resources associated with their computational advantage~\cite{ChitambarGour19}.
Entanglement has long been regarded as a central nonclassical resource in quantum information processing~\cite{horodecki2009quantum,PlenioVirmani05}, and the characterization and detection of multipartite entanglement remain active subjects of study~\cite{GuhneToth09,ma2024multipartite}.
For pure-state quantum computation, the growth of multipartite entanglement has been identified as a necessary ingredient for exponential speed-up in broad settings~\cite{JozsaLinden03,Vidal03}.
At the same time, the relation between entanglement and quantum advantage is subtle, as mixed-state models and oracle algorithms can exhibit computational advantages or nonclassical correlations even when entanglement is absent or strongly limited~\cite{KnillLaflamme98,BIHAM200415,Datta08,lanyon2008experimental}.

For instance, separable mixed-state quantum computations have been shown to perform more reliably than the best corresponding classical strategies in certain oracle problems, such as the Deutsch--Jozsa and Simon problems~\cite{BIHAM200415}.
Moreover, deterministic quantum computation with one pure qubit provides a mixed-state model in which nonclassical correlations can arise without entanglement~\cite{KnillLaflamme98,Datta08,lanyon2008experimental}.
These observations motivate a broader perspective on quantum resources, in which coherence and robustness measures also play an important role.
Quantum coherence admits a rigorous formulation as a resource~\cite{Baumgratz14,Streltsov17}.
Robustness measures, first introduced in entanglement theory~\cite{vidal1999robustness,steiner2003generalized}, have since been extended to other quantum resources, including coherence~\cite{Napoli16,Piani16}, and provide useful links between resource quantification and operational tasks, such as discrimination problems~\cite{PianiWatrous09,Napoli16,Piani16,Takagi19}.
This viewpoint is particularly suitable for probabilistic oracle algorithms, where the success of the algorithm is naturally described as an optimal guessing probability.

The Bernstein-Vazirani (BV) algorithm~\cite{Bernstein97} is one of the fundamental examples demonstrating quantum advantage in query complexity.
The goal of the algorithm is to identify an unknown $N$-bit string $\bar{a}$ encoded by a Boolean function $f_{\bar{a}}(\bar{x})=\bar{a}\cdot\bar{x}$ mod $2$ for $N$-bit string $\bar{x}$, where $\bar{s}\cdot\bar{t}=s_1t_1+s_2t_2+\cdots+s_Nt_N$ for $N$-bit strings $\bar{s}=s_1s_2\cdots s_N$ and $\bar{t}=t_1t_2\cdots t_N$.
Throughout this work, we identify each $N$-bit string $\bar{x}\in\left\{0,1\right\}^N$ with the corresponding integer $x\in\left\{0,1,\ldots,2^N-1\right\}$ according to its binary representation.
Under this identification, we write $\ket{x}$ instead of $\ket{\bar{x}}$ when no confusion can arise.
While $N$ oracle queries are generally required in the classical setting, the BV algorithm determines the hidden string with a single query.
This reduction from $N$ queries to a single query in query complexity has made the BV algorithm one of the fundamental examples illustrating the power of quantum computation.

In the standard implementation of the BV algorithm, the oracle unitary $U_{\bar{a}}$ acts on the specific initial state $\ket{-}\ket{+}^{\otimes N}$, that is, 
\[U_{\bar{a}}\ket{-}\ket{+}^{\otimes N} =  \frac{1}{\sqrt{2^N}} \ket{-} \sum_{x=0}^{2^N-1}(-1)^{\bar{a}\cdot\bar{x}}\ket{x}.\]
The action of the oracle imprints the information of the hidden string into the relative phases of the computational basis states, allowing the unknown string to be recovered by applying Hadamard gates followed by measurements in the computational basis.
A natural extension is provided by the probabilistic BV algorithm~\cite{Naseri22}, where the restriction on the initial state is removed and one allows arbitrary quantum states together with general quantum measurements.
In this setting, the success probability depends on the structure of the initial state, and the algorithm no longer achieves perfect performance in general.
Consequently, it becomes meaningful to quantify how suitable a given initial state is for solving the BV problem.
This leads to the notion of the performance of the algorithm, defined as the optimal average probability of correctly identifying the hidden string.
In particular, one may ask which classes of states achieve high performance.

In this work, we investigate a modified oracle 
\begin{align*}
    W_{\bar{a}}=\left( \sigma_z \otimes I \right) U_{\bar{a}} \left( \sigma_z \otimes I \right) U_{\bar{a}},
\end{align*}
which is constructed by sequentially applying the original oracle unitary $U_{\bar{a}}$ twice, interleaved with local phase operations on the oracle register, i.e., on the first qubit system.
Accordingly, the modified protocol requires two oracle queries, in contrast to the single-query structure of the standard BV algorithm.
However, the corresponding protocol, referred to as the initialization-free (IF) BV algorithm~\cite{Chi01}, possesses the advantage that the oracle register may be initialized in an arbitrary state without affecting the perfect identification of the hidden string, since 
\[W_{\bar{a}}\ket{\xi}\ket{+}^{\otimes N} =  \frac{1}{\sqrt{2^N}} \ket{\xi}\sum_{x=0}^{2^N-1}(-1)^{\bar{a}\cdot\bar{x}}\ket{x}\]
for any one-qubit pure state $\ket{\xi}$.
In other words, unlike the BV algorithm, the IF-BV algorithm successfully identifies the hidden string with unit probability even when an arbitrary state is employed as the oracle register.
Thus, although the IF-BV algorithm uses an additional query, it removes a significant initialization constraint present in the standard BV algorithm.

This observation naturally raises the question whether a similar advantage persists in the probabilistic setting.
More specifically, one may ask whether the probabilistic IF-BV algorithm admits a richer performance structure than the probabilistic BV algorithm, and how the performance depends on the choice of the initial state.
If the IF-BV algorithm allows a broader set of successful initializations, then it is natural to expect that its probabilistic version may also allow a broader family of initial states with high performance.
Addressing these questions is the main purpose of the present work.

We derive an explicit expression for the performance of the probabilistic IF-BV algorithm for arbitrary pure initial states and obtain a complete characterization of the states achieving maximal performance.
In addition, we provide a constructive procedure for generating such optimal states and prove that the probabilistic IF-BV algorithm outperforms the standard probabilistic BV algorithm under a suitable coefficient-ordering assumption.
These results demonstrate that the IF-BV framework enlarges the class of high-performance initial states and provides a deeper understanding of the relationship between the structure of the initial state and the performance of the algorithm.
In particular, our results show that maximal performance is governed by a balancing relation between the probability distributions associated with the structure of the initial state, rather than by a simple requirement such as maximal coherence.

\section{Probabilistic IF-BV Algorithm}
The deterministic IF-BV algorithm successfully identifies the hidden bit string with unit probability when the initial state is of the form $\ket{\xi}\ket{+}^{\otimes N}$, where the oracle register may be prepared in an arbitrary state $\ket{\xi}$.
A natural question is to ask how the algorithm behaves when this restriction on the initial state is removed.
To address this question, we consider the probabilistic IF-BV algorithm, in which arbitrary initial states and general quantum measurements are allowed.

More precisely, we assume that no prior information about the hidden bit string $\bar{a}$ is available, so that each of the $2^N$ possible strings occurs with equal probability.
For a given general initial state $\rho$, the oracle unitary $W_{\bar a}$ is applied, producing the output state $W_{\bar{a}}\rho W_{\bar{a}}^{\dagger}$.
A general quantum measurement is then performed on the output state in order to identify the value of the hidden string.
We call this protocol the probabilistic IF-BV algorithm.
Since the algorithm is no longer deterministic for general initial states, the probability of successfully identifying $\bar a$ depends on the choice of the initial state.

The probabilistic IF-BV algorithm can naturally be interpreted as a quantum channel discrimination problem~\cite{watrous2018theory,Naseri22,PianiWatrous09,Takagi19}.
Indeed, each hidden string $\bar{a}$ specifies a corresponding quantum channel through the oracle unitary $W_{\bar a}$, and the task is to determine which channel has been applied by performing a suitable measurement on the output state.
In general, the average success probability in channel discrimination is given by $\sum_i \tr \left[\Phi_i \left( \rho \right) M_i \right]$ optimized over all possible positive operator-valued measures (POVMs) $\{M_i\}$.
Motivated by this perspective, we define the performance of the probabilistic IF-BV algorithm as the optimal average probability of correctly identifying the hidden string $\bar{a}$ as follows.

\begin{Def}[Performance of the probabilistic IF-BV algorithm]
\label{Def:Performance}
For a given initial state $\rho$, the performance of the probabilistic IF-BV algorithm is defined as
\begin{align*}
    P_W(\rho)=\frac{1}{2^N}\max_{\{M_{\bar{a}}\}}\sum_{\bar{a}}\tr \left[ W_{\bar{a}}\rho W_{\bar{a}}^\dagger M_{\bar{a}} \right].
\end{align*}
\end{Def}

To analyze the performance of the probabilistic IF-BV algorithm, it is convenient to introduce a quantitative measure of coherence for the system register on the second $N$-qubit system.
Throughout this work, we employ the robustness of coherence~\cite{Streltsov17,Piani16}, a widely used coherence monotone defined with respect to the computational basis.
Although we will later see that the performance of the algorithm is not determined solely by the coherence of the initial state, the robustness of coherence provides a compact way of expressing a part of the performance formula.

Since our analysis focuses exclusively on pure initial states, it is sufficient to use the following closed expression.
For a pure state $\ket{\phi}=\sum_x C_x \ket{x}$, the robustness of coherence is given by
\begin{align*}
    R\left( \ketbra{\phi}{\phi} \right)=\left( \sum_x\abs{C_x} \right)^2 -1.
\end{align*}
In particular, the robustness of coherence vanishes for incoherent states and attains its maximum value $2^N-1$ for maximally coherent states in the computational basis.

Throughout the following discussion, we consider pure initial states of the form
\begin{align}
\label{pure initial state}
    \ket{\mu}=\alpha\ket{+}\ket{\phi_0}+\beta\ket{-}\ket{\phi_1}, \quad
    \ket{\phi_i}=\sum_{x}C_{i,x}\ket{x}
\end{align}
for computational basis $\left\{\ket{x}\right\}_{x=0,\,\cdots,\,2^{N}-1}$.

We now derive an explicit expression for the performance of the probabilistic IF-BV algorithm for arbitrary pure initial states. 
Semidefinite programming (SDP) is a standard convex optimization framework in quantum information theory, where it is used to formulate quantities such as optimal discrimination probabilities and resource measures~\cite{watrous2018theory}.
The proof of the following theorem uses this framework through the SDP characterization of the robustness of coherence, which provides a convenient way to connect the channel discrimination formulation of the probabilistic IF-BV algorithm with the coherence properties of the system register. 
The resulting formula provides the main tool for our subsequent analysis of the conditions under which maximal performance can be achieved.

\begin{Thm}
\label{Thm:performance}
    The performance of the probabilistic IF-BV algorithm for a pure initial state $\ket{\mu}$ of the form in Eq.~(\ref{pure initial state}) is given as
    \begin{align*}
        P_W(\ket{\mu})
        =& \frac{1}{2^N} \left[1
        +\abs{\alpha}^2R(\ketbra{\phi_0}{\phi_0})
        +\abs{\beta}^2R(\ketbra{\phi_1}{\phi_1})\right]\\
        &+\frac{D}{2^N},
    \end{align*}    
    where     
    \begin{align*}
        D
        &= 2\sum_{x<y}\Biggl[
        \,\prod_{z\in\{x,y\}}
        \sqrt{\abs{\alpha}^2\abs{C_{0,z}}^2+\abs{\beta}^2\abs{C_{1,z}}^2}
        \\&\qquad\qquad
        -\left(\abs{\alpha}^2\abs{C_{0,x}}\abs{C_{0,y}}+\abs{\beta}^2\abs{C_{1,x}}\abs{C_{1,y}}\right) 
        \Biggr].
    \end{align*}

\end{Thm}

\begin{proof}
    Let $d=2^N$. We first recall that the robustness of coherence with respect to the computational basis can be written as the following SDP:
    \begin{align}
    \label{eq:dual robustness}
        R(\rho)=\max\left\{\tr(\rho X)-1: X\ge 0, E(X)=I\right\},   
    \end{align}
    where $E(X)$ denotes the dephasing map in the relevant incoherent basis. 
    In the binary phase convention used here, it can be expressed as
    \begin{align*}
        E(X)=\frac{1}{d}\sum_{\bar{a}} u_{\bar{a}} X u_{\bar{a}}^\dagger\, , \quad u_{\bar a}=\sum_x \left( -1 \right)^{\bar{a} \cdot \bar{x}} \ketbra{x}{x}
    \end{align*}
    For an input state of the form $\ketbra{\xi}{\xi}\otimes\rho$, the action of $W_{\bar a}$ is equivalent to the action of $V_{\bar a}=\otimes_{i=1}^N\sigma_{z,i}^{a_i}$ on the second register.
    That is $W_{\bar{a}}\ketbra{\xi}{\xi}\otimes\rho W_{\bar{a}}^\dagger = \ketbra{\xi}{\xi} \otimes V_{\bar{a}} \rho V_{\bar{a}}^\dagger$.
    Therefore, the optimal success probability of discriminating the ensemble generated by $V_{\bar a}$ is given by
    \begin{align*}
        P_W(\ketbra{\xi}{\xi}\otimes\rho)
        &=\frac{1}{d}\max_{\left\{M_{\bar a}\right\}}\sum_{\bar a}\tr\left[V_{\bar a}\rho V_{\bar a}^{\dagger}M_{\bar a}\right]
        \\&=\frac{1+R(\rho)}{d}.
    \end{align*}

    Now consider the pure initial state $\ket{\mu}$ in Eq.~(\ref{pure initial state}).
    For each $\bar a$, let $\ket{\mu_{\bar a}}=W_{\bar a}\ket{\mu}$.
    Then
    \begin{align*}
        \ket{\mu_{\bar{a}}}=\sum_x (-1)^{{\bar{a}} \cdot x} \left( \alpha C_{0,x} \ket{+}  + \beta C_{1,x} \ket{-} \right) \ket{x}.
    \end{align*}
    Define the coefficient $C_x''=\sqrt{\abs{\alpha}^2 \abs{C_{0,x}}^2 +\abs{\beta}^2 \abs{C_{1,x}}^2}$ and  $\ket{\psi_x''}=(1/C_x'')(\alpha C_{0,x}\ket{+}+\beta C_{1,x}\ket{-})\ket{x}$ for $C_x''\neq0$.
    If $C_x''=0$, the vector $\ket{\psi_x''}$ can be chosen arbitrarily, since this component does not contribute to $\ket{\mu}$.
    Then
    \begin{align*}
    \ket{\mu_{\bar{a}}}=\sum_x (-1)^{{\bar{a}} \cdot x} C_x'' \ket{\psi_x''}\, ,
    \quad
    \ket{\mu}=\sum_x C_x''\ket{\psi_x''}.
    \end{align*}
    Hence, on the subspace $\mathcal{H}=\operatorname{span}{\ket{\psi_x''}}_x$, the action of $W_{\bar a}$ on $\ket{\mu}$ is equivalent to the action of
    \begin{align*}
        W_{\bar a}''=\sum_x(-1)^{\bar a\cdot x}\ketbra{\psi_x''}{\psi_x''}.
    \end{align*}
    Since all states $\ket{\mu_{\bar a}}$ lie in $\mathcal{H}$, it is sufficient to optimize the measurement on this subspace.
    Any POVM on $\mathcal{H}$ can be extended to the full Hilbert space without changing the success probability.
    Thus we may regard ${\ket{\psi_x''}}_x$ as the incoherent basis of the relevant subspace and denote by $R''$ the robustness of coherence with respect to this basis.

    We now show that
    \begin{align}
    \label{eq:PW Rpp}
        P_W(\ket{\mu})=\frac{1+R''(\ketbra{\mu}{\mu})}{d}.
    \end{align}
    Let $X$ be an optimal operator for the SDP in Eq.~(\ref{eq:dual robustness}) and define 
    \begin{align*}
        M_a''=\frac{1}{d}W_a''XW_a''^\dagger.
    \end{align*}
    Since $X$ is positive, so is $M_a''$.
    Moreover, using
    \begin{align*}
        I=E(X)
        =\sum_x \bra{\psi_x''}X\ket{\psi_x''}\ketbra{\psi_x''}{\psi_x''}
    \end{align*}
    together with $\bra{\psi_x''} X \ket{\psi_x''}=1$, we have
    \begin{align*}
        \sum_a M_{\bar{a}}'' 
        =\frac{1}{d} \sum_{\bar{a}} W_a'' X W_a''^\dagger
        = \sum_x \ketbra{\psi_x''}{\psi_x''} =I.
    \end{align*}
    Hence we obtain that $\{M_{\bar a}''\}$ is a POVM.
    It follows that
    \begin{align}
        \frac{1+R''\left(\ketbra{\mu}{\mu}\right)}{d}
        \nonumber
        &= \frac{1}{d} \sum_a \tr [W_a'' \ketbra{\mu}{\mu}W_a''M_a'']
        \\&\leq \max_{\left\{ M_a \right\}} \frac{1}{d} \sum_a \tr [W_a'' \ketbra{\mu}{\mu}W_a''M_a]
        \nonumber
        \\&=P_W(\ket{\mu}).
    \label{eq:lower bound}
    \end{align}

    On the other hand, let ${M_{\bar a}}$ be an arbitrary POVM on $\mathcal{H}$.
    By the definition of robustness, there exists a state $\tau$ and an incoherent state $\sigma$ with respect to the basis $\left\{\ket{\psi_x''}\right\}_x$ such that
    \begin{align*}
        \ketbra{\mu}{\mu}=\left(1+R''(\ketbra{\mu}{\mu})\right)\sigma-R''\left(\ketbra{\mu}{\mu}\right)\tau.
    \end{align*}
    Since $\sigma$ is incoherent in the basis $\{\ket{\psi_x''}\}_x$, it is invariant under every $W_{\bar a}''$, that is, $W_{\bar a}''\sigma W_{\bar a}''^{\dagger}=\sigma$. 
    Therefore,
    \begin{align*}
        \sum_{\bar{a}} \tr [W_a'' \ketbra{\mu}{\mu}W_a''^\dagger M_a]
        \leq \left( 1+R''\left(\ketbra{\mu}{\mu}\right)\right).
    \end{align*}
    Dividing both sides by $d$ and maximizing over all POVMs gives
    \begin{align}
    \label{eq:upper bound}
        P_W(\ket{\mu})\leq\frac{1+R''(\ketbra{\mu}{\mu})}{d}.
    \end{align}
    Combining Eqs.~(\ref{eq:lower bound}) and~(\ref{eq:upper bound}), we obtain Eq.~(\ref{eq:PW Rpp}).

    Since $\ket{\mu}=\sum_x C_x''\ket{\psi_x''}$ is a pure state, the robustness of coherence with respect to the basis ${\ket{\psi_x''}}$ is given by
    \begin{align*}
        R''\left(\ketbra{\mu}{\mu}\right)=\left( \sum_x \abs{C_x''} \right)^2-1
    \end{align*}    
    By the definition of $C_x''$, we have
    \begin{align*}
        1+R''&(\ketbra{\mu}{\mu})=\left(\sum_x\sqrt{\abs{\alpha}^2\abs{C_{0,x}}^2+\abs{\beta}^2\abs{C_{1,x}}^2}\right)^2,
   \end{align*}
   which is equal to 
   \begin{align*}  
\sum_x &\left(\abs{\alpha}^2\abs{C_{0,x}}^2+\abs{\beta}^2\abs{C_{1,x}}^2\right)\\      &+2\sum_{x<y}\left(
        \prod_{z\in\{x,y\}} \sqrt{\abs{\alpha}^2\abs{C_{0,z}}^2+\abs{\beta}^2\abs{C_{1,z}}^2}\right).
    \end{align*}

    The first term is equal to one because $\ket{\phi_0}$ and $\ket{\phi_1}$ are normalized and $\abs{\alpha}^2+\abs{\beta}^2=1$.
    To relate the remaining terms to the robustness of coherence of $\ket{\phi_0}$ and $\ket{\phi_1}$, we add and subtract the corresponding cross terms:
    \begin{align*}
        1+&R''\left(\ketbra{\mu}{\mu}\right) \nonumber
        \\ =&1+2\abs{\alpha}^2\sum_{x<y}\abs{C_{0,x}}\abs{C_{0,y}}+2\abs{\beta}^2\sum_{x<y}\abs{C_{1,x}}\abs{C_{1,y}}
        \\
        &+2\sum_{x<y}\Biggl[
        \,\prod_{z\in\{x,y\}}
        \sqrt{\abs{\alpha}^2\abs{C_{0,z}}^2+\abs{\beta}^2\abs{C_{1,z}}^2}
        \\&\mspace{55mu}
        -\left(\abs{\alpha}^2\abs{C_{0,x}}\abs{C_{0,y}}+\abs{\beta}^2\abs{C_{1,x}}\abs{C_{1,y}}\right) 
        \Biggr].
    \end{align*}
    For a pure state $\ket{\phi_j}=\sum_x C_{j,x}\ket{x}$, the robustness of coherence is
    \begin{align*}
        R\left(\ketbra{\phi_j}{\phi_j}\right)=2\sum_{x<y}\abs{C_{j,x}}\abs{C_{j,y}},
        \quad j=0,1.
    \end{align*}
    Therefore,
    \begin{align*}
        &1+R''\left(\ketbra{\mu}{\mu}\right)
        =1+\abs{\alpha}^2R\left(\ketbra{\phi_0}{\phi_0}\right)
        \\&\mspace{140mu}+\abs{\beta}^2R\left(\ketbra{\phi_1}{\phi_1}\right)+D,
    \end{align*}
    where
    \begin{align*}
        D
        &= 2\sum_{x<y}\Biggl[
        \,\prod_{z\in\{x,y\}}
        \sqrt{\abs{\alpha}^2\abs{C_{0,z}}^2+\abs{\beta}^2\abs{C_{1,z}}^2}
        \\&\qquad\qquad
        -\left(\abs{\alpha}^2\abs{C_{0,x}}\abs{C_{0,y}}+\abs{\beta}^2\abs{C_{1,x}}\abs{C_{1,y}}\right) 
        \Biggr].
    \end{align*}
    Here, the nonnegativity of $D$ follows directly from the Cauchy-Schwarz inequality.
    Finally, substituting this expression into Eq.~(\ref{eq:PW Rpp}), we obtain
    \begin{align*}
        P_W(\ket{\mu})
        =& \frac{1}{2^N} \left[1
        +\abs{\alpha}^2R(\ketbra{\phi_0}{\phi_0})
        +\abs{\beta}^2R(\ketbra{\phi_1}{\phi_1})\right]\\
        &+\frac{D}{2^N},
    \end{align*}  
\end{proof}

Theorem~\ref{Thm:performance} provides an explicit expression for the performance of the probabilistic IF-BV algorithm in terms of the structure of the initial state.
The formula contains contributions associated with the robustness of coherence of $\ket{\phi_0}$ and $\ket{\phi_1}$, together with an additional term $D$ that depends on both components simultaneously.
This indicates that the performance of the probabilistic IF-BV algorithm depends on the structure of the initial state.

\section{States with maximal performance}
Having obtained an explicit performance formula of the probabilistic IF-BV algorithm, we now investigate the conditions under which the probabilistic IF-BV algorithm achieves maximal performance.
While the performance formula in Theorem~\ref{Thm:performance} applies to arbitrary pure states, its structure is sufficiently complicated so that it is not immediately obvious to know under what conditions the optimal value $P_W=1$ can be attained.

We first identify some simple classes of initial states which have maximal performance, and then present a collection of examples showing that these conditions do not provide a complete characterization.
These observations naturally motivate the search for a necessary and sufficient condition for optimality, which can ultimately lead to a constructive description of all maximal-performance states.

It is worth noting that a related point also arises in the standard probabilistic BV algorithm.
Reference~\cite{Naseri22} states that the standard probabilistic BV algorithm achieves maximal performance if and only if the initial state is of the form $\ket{\mu}=\ket{-}\ket{\psi_{\max}}$ where $\ket{\psi_{\max}}$ is maximally coherent.
Our analysis supports this conclusion and clarifies the role of the coefficient ordering used in the performance formula.
In particular, once the initial state $\ket{\mu}$ is parametrized under the relabeling convention $\abs{C_{1,0}}=\max_x \abs{C_{1,x}}$, the maximal performance condition in the standard probabilistic BV algorithm can be derived rigorously.
For the detailed discussion, refer to Appendix~\ref{App:BV maximal condition}.

Even though the explicit expression obtained in Theorem~\ref{Thm:performance} involves several different contributions, there are some important classes of initial states for which maximal performance follows immediately.
In particular, when one or both of the states $\ket{\phi_0}$ and $\ket{\phi_1}$ exhibit maximal coherence, the performance formula simplifies considerably.

The following corollary provides several simple configurations of initial states that have maximal performance.
These conditions provide useful and easily verifiable sufficient criteria, which serve as a starting point for our subsequent discussion.

\begin{Cor}
\label{Cor:condition for maximal performance}
    The probabilistic IF-BV algorithm achieves maximal performance, that is, $P_W(\ket{\mu})=1$, if the initial state $\ket{\mu}$ satisfies one of the following conditions:
    \begin{itemize}
        \item $\ket{\phi_0}$ and $\ket{\phi_1}$ are both maximally coherent in the computational basis,
        \item $\alpha=0$ and $\ket{\phi_1}$ is maximally coherent in the computational basis,
        \item $\beta=0$ and $\ket{\phi_0}$ is maximally coherent in the computational basis.
    \end{itemize}
\end{Cor}

Although the conditions listed in Corollary~\ref{Cor:condition for maximal performance} establish several simple ways of obtaining maximal performance, they do not completely characterize the optimal states.
As we now demonstrate, there exist many initial states that lie outside the classes described in the corollary and nevertheless achieve maximal performance.
The following examples illustrate that the relationship between coherence and performance is substantially more subtle than one might initially expect.

\begin{Ex}
\label{Ex:Examples}
    The conditions given in Corollary~\ref{Cor:condition for maximal performance} are sufficient for achieving maximal performance, but they are not necessary.
    In fact, there exist a variety of initial states that do not satisfy any of the conditions in Corollary 3 and yet still achieve $P_W(\ket{\mu})=1$.
    We present several examples illustrating this point.

    First, consider the case $N=2$ and choose
    \begin{gather*}
        \abs{\alpha}^2=\frac{1}{4}, \qquad
        \abs{\beta}^2=\frac{3}{4}, \\
        \ket{\phi_0}=
        \frac{1}{\sqrt{10}}\ket{00}
        +\sqrt{\frac{1}{5}}\ket{01}
        +\sqrt{\frac{3}{10}}\ket{10}
        +\sqrt{\frac{2}{5}}\ket{11}, \\
        \ket{\phi_1}=
        \sqrt{\frac{3}{10}}\ket{00}
        +\sqrt{\frac{4}{15}}\ket{01}
        +\sqrt{\frac{7}{30}}\ket{10}
        +\sqrt{\frac{1}{5}}\ket{11}.
    \end{gather*}
    Then neither $\ket{\phi_0}$ nor $\ket{\phi_1}$ is incoherent or maximally coherent.
    Nevertheless, the corresponding state $\ket{\mu}$ achieves $P_W(\ket{\mu})=1$.
    This shows that the converse of Corollary~\ref{Cor:condition for maximal performance} does not hold.

    Second, let
    \begin{gather*}
        \abs{\alpha}^2=\frac{1}{8}, \quad
        \abs{\beta}^2=\frac{7}{8}, \\
        \ket{\phi_0}=\ket{11}, \\
        \ket{\phi_1}=
        \sqrt{\frac{2}{7}}\ket{00}
        +\sqrt{\frac{2}{7}}\ket{01}
        +\sqrt{\frac{2}{7}}\ket{10}
        +\frac{1}{\sqrt{7}}\ket{11}.
    \end{gather*}
    In this case, $\ket{\phi_0}$ is incoherent, while $\ket{\phi_1}$ is neither incoherent nor maximally coherent.
    However, the corresponding state $\ket{\mu}$ still achieves $P_W(\ket{\mu})=1$.
    This demonstrates that maximal performance can be achieved even when one of the states $\ket{\phi_0}$ and $\ket{\phi_1}$ is completely incoherent.

    Finally, consider the case $N=1$, and choose
    \begin{gather*}
        \abs{\alpha}^2=\abs{\beta}^2=\frac{1}{2}, \quad
        \ket{\phi_0}=\ket{1}, \quad
        \ket{\phi_1}=\ket{0}.
    \end{gather*}
    Then, while $P_W(\ket{\mu})=1$, both $\ket{\phi_0}$ and $\ket{\phi_1}$ are incoherent.
\end{Ex}

The examples above reveal that maximal performance can occur in a surprisingly broad range of situations.
In particular, the final example shows that when $N=1$, the probabilistic IF-BV algorithm may achieve perfect performance even though both $\ket{\phi_0}$ and $\ket{\phi_1}$ are completely incoherent in the computational basis.
We formalize this observation in the following corollary. 

\begin{Cor}
\label{Cor:performance for incoherent state}
    Assume that both $\ket{\phi_0}$ and $\ket{\phi_1}$ are incoherent in the computational basis. 
    If $N\geq2$, then the probabilistic IF-BV algorithm cannot achieve maximal performance, i.e., $P_W \left( \ket{\mu} \right) < 1$.
\end{Cor}

\begin{proof}
    Since $\ket{\phi_0}$ and $\ket{\phi_1}$ are incoherent in the computational basis, there exist computational basis states $\ket{x_0}$ and $\ket{x_1}$ such that $\ket{\phi_0}=\ket{x_0}$, $\ket{\phi_1}=\ket{x_1}$.
    Thus, $\abs{C_{0,x}}^2=\delta_{x,x_0}$ and $\abs{C_{1,x}}^2=\delta_{x,x_1}$.
    Using the following expression in the proof of Theorem 2,
    \begin{align*}
        P_W(\ket{\mu})=\frac{1}{2^N} \left( \sum_x\sqrt{\abs{\alpha}^2 \abs{C_{0,x}}^2+
    \abs{\beta}^2 \abs{C_{1,x}}^2} \right)^2,
    \end{align*}
    we see that the sum contains nonzero contributions from at most two computational basis states.
    
    If $x_0=x_1$, then
    \begin{align*}
        \sum_x\sqrt{\abs{\alpha}^2 \abs{C_{0,x}}^2+
        \abs{\beta}^2 \abs{C_{1,x}}^2}  =1,
    \end{align*}
    and hence $P_W(\ket{\mu})=\frac{1}{2^N}<1$.

    If $x_0\neq x_1$, then
    \begin{align*}
        \sum_x\sqrt{\abs{\alpha}^2 \abs{C_{0,x}}^2+
        \abs{\beta}^2 \abs{C_{1,x}}^2} =\abs{\alpha}+\abs{\beta}.
    \end{align*}
    Since $\abs{\alpha}^2+\abs{\beta}^2=1$, we have $\abs{\alpha}+\abs{\beta} \leq \sqrt{2}$.
    Therefore, $P_W(\ket{\mu}) \leq 2/2^N$.
    For $N\geq 2$, this implies
    \begin{align*}
        P_W(\ket{\mu})\leq \frac{1}{2^{N-1}} <1.
    \end{align*}
    Hence $P_W(\ket{\mu})< 1$ for all $N\geq 2$.
\end{proof}

The previous examples and corollary indicate that maximal performance cannot be characterized solely in terms of maximal coherence or incoherence of the states $\ket{\phi_0}$ and $\ket{\phi_1}$ in a given initial state in Eq.~(\ref{pure initial state}).
While maximal coherence certainly provides a sufficient condition, it is by no means necessary, and even highly asymmetric configurations may achieve perfect performance.

These observations suggest that the essential feature determining maximal performance must lie elsewhere.
Rather than relying on qualitative properties such as coherence or incoherence, one should seek a quantitative condition directly relating the components of the initial state.
The following theorem provides a necessary and sufficient condition for achieving $P_W \left( \ket{\mu} \right)=1 $.

\begin{Thm}
\label{Thm:maximal performance}
    A pure initial state $\ket{\mu}$ of the form in Eq.~(\ref{pure initial state}) has maximal performance of the probabilistic IF-BV algorithm, that is,
    \begin{align*}
        P_W \left( \ket{\mu} \right)=1  
    \end{align*}
    if and only if
    \begin{align*}
        \abs{\alpha}^2 \abs{C_{0,x}}^2+\abs{\beta}^2\abs{C_{1,x}}^2 =\frac{1}{2^N}
    \end{align*}
    for all $x$.
\end{Thm}

\begin{proof}
    We recall from the proof of Theorem 2 that the performance is given by
    \begin{align*}
        P_W(\ket{\mu})=\frac{1}{2^N} \left[ \sum_x \sqrt{\abs{\alpha}^2 \abs{C_{0,x}}^2+\abs{\beta }^2 \abs{C_{1,x} }^2}\right]^2.
    \end{align*}
    Let us define 
    \begin{align*}
        S_x = \sqrt{\abs{\alpha}^2 \abs{C_{0,x}}^2+\abs{\beta }^2 \abs{C_{1,x} }^2}.
    \end{align*} 
    Using this notation, the performance can be written as $ P_W\left(\ket{\mu}\right)=\left( \sum_x S_x \right)^2 /{2^N} $.
    
    On the other hand, we observe that
    \begin{align*}
        \sum_x S_x^2
        &=
        \sum_x \left( \abs{\alpha}^2 \abs{C_{0,x}}^2+\abs{\beta}^2\abs{C_{1,x}}^2 \right)
        \\&=
        \abs{\alpha}^2 \sum_x \abs{C_{0,x}}^2+\abs{\beta}^2\sum_x \abs{C_{1,x}}^2
        \\&=1.
    \end{align*}

    By applying the Cauchy-Schwarz inequality to the vectors $(1,1,\ldots,1)$ and $(S_x)_x$, and using $\sum_x S_x^2=1$, we find $\left(\sum_x S_x\right)^2\le 2^N$.
    Substituting this into the expression for $P_W(\ket{\mu})$, we get
    \begin{align*}
        P_W(\ket{\mu})
        =
        \frac{1}{2^N} \left( \sum_x S_x \right)^2
        \leq
        \frac{1}{2^N} \cdot 2^N
        =
        1.
    \end{align*}
    Hence, the maximum value of $P_W(\ket{\mu})$ is equal to $1$.
    
    We now analyze the equality condition.
    The equality in the Cauchy-Schwarz inequality holds if and only if the two vectors are linearly dependent, which in this case means that all $S_x$ must be equal. Therefore, for all $x$,
    \begin{align*}
        S_x=\frac{1}{\sqrt{2^N}}.
    \end{align*}
    Squaring both sides, we obtain
    \begin{align*}
        \abs{\alpha}^2 \abs{C_{0,x}}^2+\abs{\beta }^2 \abs{C_{1,x} }^2=\frac{1}{2^N}
    \end{align*}
    for all $x$.
    Conversely, if this condition holds for all $x$, then all $S_x$ are equal, and the Cauchy-Schwarz inequality becomes an equality.
\end{proof}

Theorem~\ref{Thm:maximal performance} completely characterizes the initial states that achieve maximal performance.
Beyond its theoretical significance, this characterization also has a practical consequence, it provides a direct method for constructing such states.

Rather than searching for optimal states on a case-by-case basis, one may systematically generate them by enforcing the condition established in Theorem~\ref{Thm:maximal performance}.
The following remark describes one possible construction procedure.

\begin{Rem}[Methods for constructing an initial state that maximizes performance]
\label{Rem:Const-maximal performance}
    We now present a general method for constructing a pure initial state that achieves maximal performance in the probabilistic IF-BV algorithm.
    For general $N \in \mathbb{N}$ and an arbitrary real number $\abs{\alpha}^2 \in (0,1)$, choose a random probability distribution $\{ p_x \}$ where $p_x \leq 1 / \left( \abs{\alpha}^2 2^N\right)$ for all $x$.
    Then for
    \begin{align*}
        q_x=\frac{\left( 1/{2^N} -\abs{\alpha}^2 p_x \right)}{1-\abs{\alpha}^2} \geq 0,
    \end{align*}
    $\{ q_x \}$ is also a probability distribution.
    Following that, $\abs{\alpha}^2 p_x+\left( 1-\abs{\alpha}^2 \right) q_x = 1/2^N$ holds for all $x$.
    Now, for each $x$, select complex numbers $C_{0,x}$ and $C_{1,x}$ satisfying $\abs{C_{0,x}}^2=p_x$ and $\abs{C_{1,x}}^2=q_x$.
    Finally, for $\ket{\mu}=\alpha\ket{+}\ket{\phi_0}+\beta\ket{-}\ket{\phi_1}$ as in Eq.~(\ref{pure initial state}), the performance of the probabilistic IF-BV algorithm $P_W(\ket{\mu})$ has maximum value.
\end{Rem}

This construction shows that the class of initial states with maximal performance is remarkably broad.
In particular, maximal performance is determined not by maximal coherence, but by a balancing relation between the two states $\ket{\phi_0}$ and $\ket{\phi_1}$ in Eq.~(\ref{pure initial state}).
This highlights the richer and more flexible performance structure of the probabilistic IF-BV algorithm compared to the standard probabilistic BV algorithm.

\section{Comparison with The Standard Probabilistic BV Algorithm}
We next compare the probabilistic IF-BV algorithm with the standard probabilistic BV algorithm, focusing on how the modified oracle affects the optimal success probability for a given pure initial state.
Since the IF-BV algorithm succeeds for a broader class of initial states than the standard BV algorithm, it is natural to ask whether a similar advantage persists in the probabilistic setting.

More specifically, one may wonder whether the additional flexibility provided by the modified oracle can also lead to an improvement in the optimal average success probability. The comparison below is stated under an explicit ordering assumption on the coefficients of the initial state. Namely, we assume that the distinguished computational basis element $x=0$, which appears separately in the performance formula for the standard probabilistic BV algorithm, corresponds to a coefficient of $\ket{\phi_1}$ with maximal modulus. Under this assumption, the following theorem establishes that the probabilistic IF-BV algorithm outperforms the standard probabilistic BV algorithm.

\begin{Thm}
\label{Thm:PW>=P}
    Assume that the pure initial state $\ket{\mu}$ is of the form $\ket{\mu}=\alpha\ket{+}\ket{\phi_0}+\beta\ket{-}\ket{\phi_1}$, where $\ket{\phi_i}=\sum_x C_{i,x}\ket{x}$.
    We relabel the computational basis, if necessary, so that $\abs{C_{1,0}}=\max_x \abs{C_{1,x}}$.    
    Let $P(\ket{\mu})$ be the performance of the probabilistic BV algorithm for a pure initial state $\ket{\mu}$, that is, 
    \[
     P(\rho)=\frac{1}{2^N}\max_{\{M_{\bar{a}}\}}\sum_{\bar{a}}\tr \left[ U_{\bar{a}}\rho U_{\bar{a}}^\dagger M_{\bar{a}} \right] 
    \]
    for an initial state $\rho$.
    Then, under the given relabeling convention, $P_W(\ket{\mu}) \geq P(\ket{\mu}) $ for any pure initial state $\ket{\mu}$.
\end{Thm}

\begin{proof}
    According to the proof of Theorem~\ref{Thm:performance}, the performance of the probabilistic IF-BV algorithm can be written as 
    \begin{align*}
        P_W(\ket{\mu})=\frac{1}{2^N} \left[ \sum_x \sqrt{\abs{\alpha}^2 \abs{C_{0,x}}^2+\abs{\beta }^2 \abs{C_{1,x} }^2}\right]^2.
    \end{align*}
    On the other hand, the performance of the probabilistic BV algorithm is given by
    \begin{align*}
        P(\ket{\mu})=\frac{1}{2^N} \left[  \sqrt{\abs{\alpha}^2 +\abs{\beta }^2 \abs{C_{1,0} }^2} + \abs{\beta } \sum_{x\neq0}\abs{C_{1,x}} \right]^2. 
    \end{align*}

    To prove that $P_W(\ket{\mu}) \geq P(\ket{\mu})$, define $a=\abs{\alpha}^2$ and $b=\abs{\beta}^2$, so that $a+b=1$, and denote $p_x=\abs{C_{0,x}}^2$ and $q_x=\abs{C_{1,x}}^2$.
    Then both $\{p_x\}$ and $\{q_x\}$ are probability distributions.
    
    Now for a given probability distribution $\{q_x\}$, consider the function 
    \begin{align*}
        G\left( \{p_x\} \right)=\sum_x \sqrt{ap_x + bq_x}.
    \end{align*}
    Then we have $P_W \left( \ket{\mu} \right) = 1/2^N G \left( \{ p_x \} \right)^2$.
    We analyze the minimum of $G \left( \{ p_x \} \right)$ over all probability distributions $\{p_x\}$. 
    Each term $\sqrt{ap_x + bq_x}$ is a concave function of $p_x$, and therefore $G$ is concave in the vector $\{p_x\}$. 
    Since the probability simplex is convex, the minimum of a concave function over this domain is attained at an extreme point.
    Hence, the minimum occurs when all the weight is concentrated on a single index, i.e., there exists $j$ such that $p_i=\delta_{i,j}$.
    Substituting this into $G$, we obtain
    \begin{align*}
        G\left( \{p_x \} \right) 
        \geq 
        \min_j \left[ \sqrt{a+bq_j} + \sum_{x \neq j} \sqrt{bq_x} \right].
    \end{align*}
    Since $\sqrt{bq_x}=\abs{\beta}\abs{C_{1,x}}$, this becomes
    \begin{align*}
        G\left( \{p_x \} \right) 
        \geq
        \min_j \left[ \sqrt{\abs{\alpha}^2+\abs{\beta}^2\abs{C_{1,j}}^2} + \abs{\beta} \sum_{x \neq j} \abs{C_{1,x}} \right].
    \end{align*}
   
   To determine which $j$ minimizes this expression, define $h(t)=\sqrt{\abs{\alpha}^2 + \abs{\beta}^2 t^2}-\abs{\beta}t$.
   If \(\alpha=0\), the inequality is immediate.
   Assume \(\alpha\neq0\).
   Then a direct computation shows that \(h'(t)<0\) for all \(t>0\).  
   By our assumption, this occurs at $j=0$.

   Therefore, we obtain
    \begin{align*}
        P_W(\ket{\mu}) 
        &\geq
        \frac{1}{2^N} \left[ \sqrt{\abs{\alpha}^2+\abs{\beta}^2\abs{C_{1,0}}^2} + \abs{\beta} \sum_{x \neq 0} \abs{C_{1,x}} \right]^2 
        \\&= 
        P(\ket{\mu}).
    \end{align*}
   
\end{proof}

This result shows that the IF-BV algorithm enlarges the class of initial states capable of achieving high performance.
In particular, the additional interaction between the terms of $\ket{\phi_0}$ and $\ket{\phi_1}$ in Eq.~(\ref{pure initial state}) enabled by the modified oracle provides a systematic performance advantage over the probabilistic BV algorithm.

It is worth emphasizing that the comparison in Theorem~\ref{Thm:PW>=P} relies on the relabeling convention imposed on the coefficients of $\ket{\phi_1}$.
Without this assumption, the inequality $P_W(\ket{\mu})\ge P(\ket{\mu})$ does not necessarily hold.
The following remark presents a simple counterexample illustrating this point.

\begin{Rem}
The convention $\abs{C_{1,0}}=\max_x \abs{C_{1,x}}$ is essential in Theorem~\ref{Thm:PW>=P}.
Without this assumption, the inequality $P_W(\ket{\mu})\ge P(\ket{\mu})$ does not necessarily hold.

As a simple counterexample, consider the case
\begin{gather*}
    N=1, \quad \abs{\alpha}^2=\abs{\beta}^2=\frac12, \\
    \ket{\phi_0}=\ket1,\quad
    \ket{\phi_1}=\frac12\ket0+\frac{\sqrt3}{2}\ket1.
\end{gather*}
Since $\abs{C_{1,0}}=1/2 <{\sqrt3}/{2}=\abs{C_{1,1}}$, the relabeling convention of Theorem~\ref{Thm:PW>=P} is violated.
A direct calculation shows that
\begin{align*}
    P_W(\ket{\mu})=\frac12\left(\sqrt{\frac18}+\sqrt{\frac78}\right)^2=\frac{4+\sqrt7}{8},
\end{align*}
whereas
\begin{align*}
    P(\ket{\mu})=\frac12\left(\sqrt{\frac58}+\sqrt{\frac38}\right)^2=\frac{4+\sqrt{15}}{8}.
\end{align*}
Hence, $P_W(\ket{\mu})<P(\ket{\mu})$.

This example shows that the ordering assumption is necessary for the comparison stated in Theorem~\ref{Thm:PW>=P}.
\end{Rem}

Although the probabilistic IF-BV algorithm outperforms its standard counterpart under the relabeling convention, there are special situations in which the two algorithms exhibit exactly the same performance.
It is therefore natural to ask under what conditions the advantage provided by the modified oracle disappears.
The following corollary gives a complete characterization of the equality case.

\begin{Cor}
\label{Cor:equality}
    Assume that the pure initial state $\ket{\mu}$ is of the form $\ket{\mu}=\alpha\ket{+}\ket{\phi_0}+\beta\ket{-}\ket{\phi_1}$, where $\ket{\phi_i}=\sum_x C_{i,x}\ket{x}$, and that the computational basis has been relabeled, if necessary, so that $\abs{C_{1,0}}=\max_x \abs{C_{1,x}}$.  
    Then the equality $P_W\left(\ket{\mu}\right)=P\left(\ket{\mu}\right)$ holds if and only if one of the following conditions is satisfied:
    \begin{itemize}
        \item $\alpha=0$,
        \item $\beta=0$ and $\ket{\phi_0}$ is incoherent in the computational basis,
        \item $\alpha\beta\neq0$ and $\ket{\phi_0}$ is supported on a single computational basis state $\ket{j}$ satisfying $\abs{C_{1,j}}=\abs{C_{1,0}}$.
    \end{itemize}    
    In particular, when $\alpha\beta\neq0$ and the maximum value $\abs{C_{1,0}}=\max_x \abs{C_{1,x}}$ is unique, the equality holds if and only if $\ket{\phi_0}=\ket{0}$.
\end{Cor}

\begin{proof}
    From the proof of Theorem~\ref{Thm:PW>=P}, the inequality $P_W(\ket{\mu}) \ge P(\ket{\mu})$ follows from $G(\{p_x\}) \ge \min_j G\left(e_j\right)$, where $e_j$ denotes the extreme point of the probability simplex satisfying $p_j=1$ and $p_i=0$ for $i\neq j$.

    If $\alpha=0$, then $G(\{p_x\})$ is independent of $\{p_x\}$, and hence $P_W(\ket{\mu})=P(\ket{\mu})$.
    Next, assume that $\beta=0$. Then $P_W(\ket{\mu})=\left(\sum_x \abs{C_{0,x}}\right)^2 /2^N$, while $P(\ket{\mu})=1/2^N$.
    Hence $P_W(\ket{\mu})=P(\ket{\mu})$ if and only if $\left(\sum_x \abs{C_{0,x}}\right)^2=1$, that is, $C_{0,j}=1$ for some $j$ and $C_{0,i}=0$ for all $i\neq j$.
    This is equivalent to $\ket{\phi_0}$ being incoherent in the computational basis.
    
    It remains to consider the case $\alpha\beta\neq0$.
    Since $G(\{p_x\})$ is strictly concave in $\{p_x\}$, the equality holds only when $\{p_x\}$ is an extreme point attaining the above minimum.
    This is equivalent to $\ket{\phi_0}$ being supported on a single computational basis state $\ket{j}$ satisfying $\abs{C_{1,j}}=\max_x \abs{C_{1,x}}$. 
\end{proof}

\section{Conclusion}
In order to investigate the performance structure of the probabilistic IF-BV algorithm, we derived an explicit formula for the performance, and identified a necessary and sufficient condition for achieving maximal performance.
Moreover, under the relabeling convention on the coefficients of $\ket{\phi_1}$, we established that the IF-BV algorithm outperforms as the probabilistic BV algorithm.

Our results show that maximal performance is not determined solely by the coherence properties of a given initial state.
Instead, the key quantity depends on the weighted distribution generated by the combination of the two states comprising the initial state.
This reveals that the IF-BV algorithm possesses a substantially broader class of optimal initial states compared to the standard probabilistic BV algorithm.

Overall, our results provide a deeper understanding of how the structure of initial states influences probabilistic quantum query algorithms.
We expect that the characterization developed in this work may be useful for analyzing more general oracle-based quantum algorithms and for understanding the role of interactions between the states composing a given initial state in probabilistic quantum computation.

\begin{acknowledgments}
This research was supported by the education and training program of the Quantum Information Research Support Center, funded through the National research foundation of Korea(NRF) by the Ministry of science and ICT(MSIT) of the Korean government (No.RS-2023-NR057243).
A.S. acknowledges support from the National Science Centre Poland (Grant No. 2022/46/E/ST2/00115).
S.L. acknowledges support from 
the Institute of Information \& Communications Technology Planning \& Evaluation (IITP) grant funded by the Ministry of Science and ICT (MSIT) (No. RS-2025-02304540), 
the National Research Foundation of Korea (NRF) of Korea grants funded by the MSIT  (No. RS-2024-00432214 and No. RS-2022-NR068791) and Creation of the Quantum Information Science R\&D Ecosystem (No. RS-2023-NR068116) through the NRF funded by the MSIT.
\end{acknowledgments}

\appendix

\section{States with maximal performance in the standard probabilistic BV algorithm}
\label{App:BV maximal condition}

In this appendix, we revisit the condition for achieving maximal performance in the standard probabilistic BV algorithm.
Reference~\cite{Naseri22} states that maximal performance $P(\ket{\mu})=1$ is achievable only if the initial state is of the form $\ket{\mu}=\ket{-}\ket{\psi_{\max}}$, where $\ket{\psi_{\max}}$ is a maximally coherent state in the computational basis.
While examining the proof presented in the Supplemental Material of Ref.~\cite{Naseri22}, we observed that the derivation implicitly requires an ordering condition on the coefficients of $\ket{\phi_1}$.
Motivated by the analysis in the present work, we therefore adopt the assumption $\abs{C_{1,0}}=\max_x \abs{C_{1,x}}$, where $\ket{\phi_1}=\sum_x C_{1,x}\ket{x}$.

Recall that a general pure initial state can be written as $\ket{\mu}=\alpha\ket{+}\ket{\phi_0}+\beta\ket{-}\ket{\phi_1}$, where $\ket{\phi_j}=\sum_x C_{j,x}\ket{x}$.
For the standard probabilistic BV algorithm, the performance is given by
\begin{align*}
     P(\ket{\mu})=\frac{1}{2^N}\left(\sqrt{\abs{\alpha}^2+\abs{\beta}^2\abs{C_{1,0}}^2}+\abs{\beta}\sum_{x\neq0}\abs{C_{1,x}}\right)^2 .
\end{align*}

\begin{Prop}
\label{Prop:BV maximal condition}
    Assume that the pure initial state $\ket{\mu}$ is of the form $\ket{\mu}=\alpha\ket{+}\ket{\phi_0}+\beta\ket{-}\ket{\phi_1}$, where $\ket{\phi_i}=\sum_x C_{i,x}\ket{x}$, and that the computational basis has been relabeled, if necessary, so that $\abs{C_{1,0}}=\max_x \abs{C_{1,x}}$.  
    Then the standard probabilistic BV algorithm achieves maximal performance, $P(\ket{\mu})=1$, if and only if $\ket{\mu}=\ket{-}\ket{\psi_{\max}}$, where $\ket{\psi_{\max}}$ is a maximally coherent state in the computational basis.
\end{Prop}

\begin{proof}
    As in Ref.~\cite{Naseri22}, the relevant final states can be expressed in a suitable orthonormal basis.
    In that basis, maximal performance $P(\ket{\mu})=1$ requires the corresponding coefficients to have equal modulus.
    Hence, if $P(\ket{\mu})=1$, then   
\begin{align}
\label{eq:43}
    \abs{\beta}^2 \abs{C_{1,x}}^2 =\frac{1}{2^N} \quad \forall x \neq 0,\\
\label{eq:44}
    \abs{\alpha}^2+\abs{\beta}^2\abs{C_{1,0}}^2=\frac{1}{2^N}.
\end{align}
From Eq.~(\ref{eq:43}), since the computational basis is ordered such that
$\abs{C_{1,0}} \ge \abs{C_{1,x}}$ for all $x$ and $\abs{\beta}^2\le1$, we obtain $\abs{C_{1,0}}^2\ge{1}/{2^N}$.
On the other hand, since $\abs{\alpha}^2=1-\abs{\beta}^2$, Eq.~(\ref{eq:44}) can be rewritten as
\begin{align*}
    \abs{\beta}^2\left( 1-\abs{C_{1,0}}^2 \right)=1-\frac{1}{2^N}.
\end{align*}
Since $\abs{\beta}^2\le1$, this implies $\abs{C_{1,0}}^2\le{1}/{2^N}$.
Combining the two inequalities, we obtain $\abs{C_{1,0}}^2={1}/{2^N}$.
Also, substituting this into $\abs{\beta}^2 \abs{C_{1,0}}^2 \ge {1}/{2^N}$, we conclude that $\abs{\beta}^2=1$.
Finally, $\ket{\phi_1}$ is maximally coherent in the computational basis and $\ket{\mu}=\ket{-}\ket{\psi_{\max}}$, where $\ket{\psi_{\max}}$ is a maximally coherent state in the computational basis.

Conversely, if $\ket{\mu}=\ket{-}\ket{\psi_{\max}}$, then a direct substitution into the performance formula gives $P(\ket{\mu})=1$.
Therefore, $P(\ket{\mu})=1$ if and only if $\ket{\mu}=\ket{-}\ket{\psi_{\max}}$, where $\ket{\psi_{\max}}$ is a maximally coherent state in the computational basis.
\end{proof}

Proposition~\ref{Prop:BV maximal condition} shows that, under the relabeling convention on the coefficients of $\ket{\phi_1}$, the maximal performance condition stated in Ref.~\cite{Naseri22} is valid.
In other words, once the largest coefficient of $\ket{\phi_1}$ is assigned to the distinguished computational basis element appearing in the performance formula, maximal performance forces the initial state to be of the form $\ket{-}\ket{\psi_{\max}}$.

However, the relabeling convention is not merely a technical detail.
If this convention is not imposed, the above characterization can fail.
The following remark illustrates this point with an explicit example.

\begin{Rem}
\label{Rem:BV maximal condition}
The relabeling convention in Proposition~\ref{Prop:BV maximal condition} is essential.
If this convention is not imposed, the conclusion of the proposition does not necessarily hold.

As an explicit example, consider the case $N=2$ and choose
\begin{gather*}
    \abs{\alpha}^2=\frac{1}{8}, \quad
    \abs{\beta}^2=\frac{7}{8}, \\
    \ket{\phi_0}=\ket{00}, \\
    \ket{\phi_1}=
    \frac{1}{\sqrt{7}}\ket{00}
    +\sqrt{\frac{2}{7}}\ket{01}
    +\sqrt{\frac{2}{7}}\ket{10}
    +\sqrt{\frac{2}{7}}\ket{11}.
\end{gather*}
In this case, the convention $\abs{C_{1,0}}=\max_x \abs{C_{1,x}}$ is violated.

Nevertheless, substituting the above coefficients into the performance formula gives
\begin{align*}
    P(\ket{\mu})
    =\frac{1}{4}\left(\sqrt{\frac{1}{8}+\frac{7}{8}\cdot\frac{1}{7}}
    +\sqrt{\frac{7}{8}}\left(3\sqrt{\frac{2}{7}}\right)\right)^2 =1.
\end{align*}
However, this state is not of the form $\ket{-}\ket{\psi_{\max}}$, since $\abs{\alpha}^2=1/8\neq0$ and $\ket{\phi_1}$ is not maximally coherent.
This example shows that the ordering condition is necessary for the maximal-performance characterization in Proposition~\ref{Prop:BV maximal condition}. 
\end{Rem}

Remark~\ref{Rem:BV maximal condition} demonstrates that the relabeling convention plays an essential role in the maximal-performance characterization of the standard probabilistic BV algorithm.
Without this convention, there exist maximal-performance states that are not of the form $\ket{-}\ket{\psi_{\max}}$.
Thus, the statement in Ref.~\cite{Naseri22} should be understood as valid under the relabeling convention specified above.

\bibliography{references}

\begin{thebibliography}{23}%
\makeatletter
\providecommand \@ifxundefined [1]{%
 \@ifx{#1\undefined}
}%
\providecommand \@ifnum [1]{%
 \ifnum #1\expandafter \@firstoftwo
 \else \expandafter \@secondoftwo
 \fi
}%
\providecommand \@ifx [1]{%
 \ifx #1\expandafter \@firstoftwo
 \else \expandafter \@secondoftwo
 \fi
}%
\providecommand \natexlab [1]{#1}%
\providecommand \enquote  [1]{``#1''}%
\providecommand \bibnamefont  [1]{#1}%
\providecommand \bibfnamefont [1]{#1}%
\providecommand \citenamefont [1]{#1}%
\providecommand \href@noop [0]{\@secondoftwo}%
\providecommand \href [0]{\begingroup \@sanitize@url \@href}%
\providecommand \@href[1]{\@@startlink{#1}\@@href}%
\providecommand \@@href[1]{\endgroup#1\@@endlink}%
\providecommand \@sanitize@url [0]{\catcode `\\12\catcode `\$12\catcode `\&12\catcode `\#12\catcode `\^12\catcode `\_12\catcode `\%12\relax}%
\providecommand \@@startlink[1]{}%
\providecommand \@@endlink[0]{}%
\providecommand \url  [0]{\begingroup\@sanitize@url \@url }%
\providecommand \@url [1]{\endgroup\@href {#1}{\urlprefix }}%
\providecommand \urlprefix  [0]{URL }%
\providecommand \Eprint [0]{\href }%
\providecommand \doibase [0]{https://doi.org/}%
\providecommand \selectlanguage [0]{\@gobble}%
\providecommand \bibinfo  [0]{\@secondoftwo}%
\providecommand \bibfield  [0]{\@secondoftwo}%
\providecommand \translation [1]{[#1]}%
\providecommand \BibitemOpen [0]{}%
\providecommand \bibitemStop [0]{}%
\providecommand \bibitemNoStop [0]{.\EOS\space}%
\providecommand \EOS [0]{\spacefactor3000\relax}%
\providecommand \BibitemShut  [1]{\csname bibitem#1\endcsname}%
\let\auto@bib@innerbib\@empty
\bibitem [{\citenamefont {Chitambar}\ and\ \citenamefont {Gour}(2019)}]{ChitambarGour19}%
  \BibitemOpen
  \bibfield  {author} {\bibinfo {author} {\bibfnamefont {E.}~\bibnamefont {Chitambar}}\ and\ \bibinfo {author} {\bibfnamefont {G.}~\bibnamefont {Gour}},\ }\bibfield  {title} {\bibinfo {title} {Quantum resource theories},\ }\href@noop {} {\bibfield  {journal} {\bibinfo  {journal} {Reviews of modern physics}\ }\textbf {\bibinfo {volume} {91}},\ \bibinfo {pages} {025001} (\bibinfo {year} {2019})}\BibitemShut {NoStop}%
\bibitem [{\citenamefont {Horodecki}\ \emph {et~al.}(2009)\citenamefont {Horodecki}, \citenamefont {Horodecki}, \citenamefont {Horodecki},\ and\ \citenamefont {Horodecki}}]{horodecki2009quantum}%
  \BibitemOpen
  \bibfield  {author} {\bibinfo {author} {\bibfnamefont {R.}~\bibnamefont {Horodecki}}, \bibinfo {author} {\bibfnamefont {P.}~\bibnamefont {Horodecki}}, \bibinfo {author} {\bibfnamefont {M.}~\bibnamefont {Horodecki}},\ and\ \bibinfo {author} {\bibfnamefont {K.}~\bibnamefont {Horodecki}},\ }\bibfield  {title} {\bibinfo {title} {Quantum entanglement},\ }\href@noop {} {\bibfield  {journal} {\bibinfo  {journal} {Reviews of modern physics}\ }\textbf {\bibinfo {volume} {81}},\ \bibinfo {pages} {865} (\bibinfo {year} {2009})}\BibitemShut {NoStop}%
\bibitem [{\citenamefont {Plenio}\ and\ \citenamefont {Virmani}(2005)}]{PlenioVirmani05}%
  \BibitemOpen
  \bibfield  {author} {\bibinfo {author} {\bibfnamefont {M.~B.}\ \bibnamefont {Plenio}}\ and\ \bibinfo {author} {\bibfnamefont {S.}~\bibnamefont {Virmani}},\ }\bibfield  {title} {\bibinfo {title} {An introduction to entanglement measures},\ }\href@noop {} {\bibfield  {journal} {\bibinfo  {journal} {arXiv preprint quant-ph/0504163}\ } (\bibinfo {year} {2005})}\BibitemShut {NoStop}%
\bibitem [{\citenamefont {G{\"u}hne}\ and\ \citenamefont {T{\'o}th}(2009)}]{GuhneToth09}%
  \BibitemOpen
  \bibfield  {author} {\bibinfo {author} {\bibfnamefont {O.}~\bibnamefont {G{\"u}hne}}\ and\ \bibinfo {author} {\bibfnamefont {G.}~\bibnamefont {T{\'o}th}},\ }\bibfield  {title} {\bibinfo {title} {Entanglement detection},\ }\href@noop {} {\bibfield  {journal} {\bibinfo  {journal} {Physics Reports}\ }\textbf {\bibinfo {volume} {474}},\ \bibinfo {pages} {1} (\bibinfo {year} {2009})}\BibitemShut {NoStop}%
\bibitem [{\citenamefont {Ma}\ \emph {et~al.}(2025)\citenamefont {Ma}, \citenamefont {Li},\ and\ \citenamefont {Shang}}]{ma2024multipartite}%
  \BibitemOpen
  \bibfield  {author} {\bibinfo {author} {\bibfnamefont {M.}~\bibnamefont {Ma}}, \bibinfo {author} {\bibfnamefont {Y.}~\bibnamefont {Li}},\ and\ \bibinfo {author} {\bibfnamefont {J.}~\bibnamefont {Shang}},\ }\bibfield  {title} {\bibinfo {title} {Multipartite entanglement measures: A review},\ }\href@noop {} {\bibfield  {journal} {\bibinfo  {journal} {Fundamental Research}\ }\textbf {\bibinfo {volume} {5}},\ \bibinfo {pages} {2489} (\bibinfo {year} {2025})}\BibitemShut {NoStop}%
\bibitem [{\citenamefont {Jozsa}\ and\ \citenamefont {Linden}(2003)}]{JozsaLinden03}%
  \BibitemOpen
  \bibfield  {author} {\bibinfo {author} {\bibfnamefont {R.}~\bibnamefont {Jozsa}}\ and\ \bibinfo {author} {\bibfnamefont {N.}~\bibnamefont {Linden}},\ }\bibfield  {title} {\bibinfo {title} {On the role of entanglement in quantum-computational speed-up},\ }\href@noop {} {\bibfield  {journal} {\bibinfo  {journal} {Proceedings of the Royal Society of London. Series A: Mathematical, Physical and Engineering Sciences}\ }\textbf {\bibinfo {volume} {459}},\ \bibinfo {pages} {2011} (\bibinfo {year} {2003})}\BibitemShut {NoStop}%
\bibitem [{\citenamefont {Vidal}(2003)}]{Vidal03}%
  \BibitemOpen
  \bibfield  {author} {\bibinfo {author} {\bibfnamefont {G.}~\bibnamefont {Vidal}},\ }\bibfield  {title} {\bibinfo {title} {Efficient classical simulation of slightly entangled quantum computations},\ }\href@noop {} {\bibfield  {journal} {\bibinfo  {journal} {Physical review letters}\ }\textbf {\bibinfo {volume} {91}},\ \bibinfo {pages} {147902} (\bibinfo {year} {2003})}\BibitemShut {NoStop}%
\bibitem [{\citenamefont {Knill}\ and\ \citenamefont {Laflamme}(1998)}]{KnillLaflamme98}%
  \BibitemOpen
  \bibfield  {author} {\bibinfo {author} {\bibfnamefont {E.}~\bibnamefont {Knill}}\ and\ \bibinfo {author} {\bibfnamefont {R.}~\bibnamefont {Laflamme}},\ }\bibfield  {title} {\bibinfo {title} {Power of one bit of quantum information},\ }\href@noop {} {\bibfield  {journal} {\bibinfo  {journal} {Physical Review Letters}\ }\textbf {\bibinfo {volume} {81}},\ \bibinfo {pages} {5672} (\bibinfo {year} {1998})}\BibitemShut {NoStop}%
\bibitem [{\citenamefont {Biham}\ \emph {et~al.}(2004)\citenamefont {Biham}, \citenamefont {Brassard}, \citenamefont {Kenigsberg},\ and\ \citenamefont {Mor}}]{BIHAM200415}%
  \BibitemOpen
  \bibfield  {author} {\bibinfo {author} {\bibfnamefont {E.}~\bibnamefont {Biham}}, \bibinfo {author} {\bibfnamefont {G.}~\bibnamefont {Brassard}}, \bibinfo {author} {\bibfnamefont {D.}~\bibnamefont {Kenigsberg}},\ and\ \bibinfo {author} {\bibfnamefont {T.}~\bibnamefont {Mor}},\ }\bibfield  {title} {\bibinfo {title} {Quantum computing without entanglement},\ }\href@noop {} {\bibfield  {journal} {\bibinfo  {journal} {Theoretical Computer Science}\ }\textbf {\bibinfo {volume} {320}},\ \bibinfo {pages} {15} (\bibinfo {year} {2004})}\BibitemShut {NoStop}%
\bibitem [{\citenamefont {Datta}\ \emph {et~al.}(2008)\citenamefont {Datta}, \citenamefont {Shaji},\ and\ \citenamefont {Caves}}]{Datta08}%
  \BibitemOpen
  \bibfield  {author} {\bibinfo {author} {\bibfnamefont {A.}~\bibnamefont {Datta}}, \bibinfo {author} {\bibfnamefont {A.}~\bibnamefont {Shaji}},\ and\ \bibinfo {author} {\bibfnamefont {C.~M.}\ \bibnamefont {Caves}},\ }\bibfield  {title} {\bibinfo {title} {Quantum discord and the power of one qubit},\ }\href@noop {} {\bibfield  {journal} {\bibinfo  {journal} {Physical review letters}\ }\textbf {\bibinfo {volume} {100}},\ \bibinfo {pages} {050502} (\bibinfo {year} {2008})}\BibitemShut {NoStop}%
\bibitem [{\citenamefont {Lanyon}\ \emph {et~al.}(2008)\citenamefont {Lanyon}, \citenamefont {Barbieri}, \citenamefont {Almeida},\ and\ \citenamefont {White}}]{lanyon2008experimental}%
  \BibitemOpen
  \bibfield  {author} {\bibinfo {author} {\bibfnamefont {B.~P.}\ \bibnamefont {Lanyon}}, \bibinfo {author} {\bibfnamefont {M.}~\bibnamefont {Barbieri}}, \bibinfo {author} {\bibfnamefont {M.~P.}\ \bibnamefont {Almeida}},\ and\ \bibinfo {author} {\bibfnamefont {A.~G.}\ \bibnamefont {White}},\ }\bibfield  {title} {\bibinfo {title} {Experimental quantum computing without entanglement},\ }\href@noop {} {\bibfield  {journal} {\bibinfo  {journal} {Physical review letters}\ }\textbf {\bibinfo {volume} {101}},\ \bibinfo {pages} {200501} (\bibinfo {year} {2008})}\BibitemShut {NoStop}%
\bibitem [{\citenamefont {Baumgratz}\ \emph {et~al.}(2014)\citenamefont {Baumgratz}, \citenamefont {Cramer},\ and\ \citenamefont {Plenio}}]{Baumgratz14}%
  \BibitemOpen
  \bibfield  {author} {\bibinfo {author} {\bibfnamefont {T.}~\bibnamefont {Baumgratz}}, \bibinfo {author} {\bibfnamefont {M.}~\bibnamefont {Cramer}},\ and\ \bibinfo {author} {\bibfnamefont {M.~B.}\ \bibnamefont {Plenio}},\ }\bibfield  {title} {\bibinfo {title} {Quantifying coherence},\ }\href@noop {} {\bibfield  {journal} {\bibinfo  {journal} {Physical Review Letters}\ }\textbf {\bibinfo {volume} {113}},\ \bibinfo {pages} {140401} (\bibinfo {year} {2014})}\BibitemShut {NoStop}%
\bibitem [{\citenamefont {Streltsov}\ \emph {et~al.}(2017)\citenamefont {Streltsov}, \citenamefont {Adesso},\ and\ \citenamefont {Plenio}}]{Streltsov17}%
  \BibitemOpen
  \bibfield  {author} {\bibinfo {author} {\bibfnamefont {A.}~\bibnamefont {Streltsov}}, \bibinfo {author} {\bibfnamefont {G.}~\bibnamefont {Adesso}},\ and\ \bibinfo {author} {\bibfnamefont {M.~B.}\ \bibnamefont {Plenio}},\ }\bibfield  {title} {\bibinfo {title} {Colloquium: Quantum coherence as a resource},\ }\href@noop {} {\bibfield  {journal} {\bibinfo  {journal} {Reviews of Modern Physics}\ }\textbf {\bibinfo {volume} {89}},\ \bibinfo {pages} {041003} (\bibinfo {year} {2017})}\BibitemShut {NoStop}%
\bibitem [{\citenamefont {Vidal}\ and\ \citenamefont {Tarrach}(1999)}]{vidal1999robustness}%
  \BibitemOpen
  \bibfield  {author} {\bibinfo {author} {\bibfnamefont {G.}~\bibnamefont {Vidal}}\ and\ \bibinfo {author} {\bibfnamefont {R.}~\bibnamefont {Tarrach}},\ }\bibfield  {title} {\bibinfo {title} {Robustness of entanglement},\ }\href@noop {} {\bibfield  {journal} {\bibinfo  {journal} {Physical Review A}\ }\textbf {\bibinfo {volume} {59}},\ \bibinfo {pages} {141} (\bibinfo {year} {1999})}\BibitemShut {NoStop}%
\bibitem [{\citenamefont {Steiner}(2003)}]{steiner2003generalized}%
  \BibitemOpen
  \bibfield  {author} {\bibinfo {author} {\bibfnamefont {M.}~\bibnamefont {Steiner}},\ }\bibfield  {title} {\bibinfo {title} {Generalized robustness of entanglement},\ }\href@noop {} {\bibfield  {journal} {\bibinfo  {journal} {Physical Review A}\ }\textbf {\bibinfo {volume} {67}},\ \bibinfo {pages} {054305} (\bibinfo {year} {2003})}\BibitemShut {NoStop}%
\bibitem [{\citenamefont {Napoli}\ \emph {et~al.}(2016)\citenamefont {Napoli}, \citenamefont {Bromley}, \citenamefont {Cianciaruso}, \citenamefont {Piani}, \citenamefont {Johnston},\ and\ \citenamefont {Adesso}}]{Napoli16}%
  \BibitemOpen
  \bibfield  {author} {\bibinfo {author} {\bibfnamefont {C.}~\bibnamefont {Napoli}}, \bibinfo {author} {\bibfnamefont {T.~R.}\ \bibnamefont {Bromley}}, \bibinfo {author} {\bibfnamefont {M.}~\bibnamefont {Cianciaruso}}, \bibinfo {author} {\bibfnamefont {M.}~\bibnamefont {Piani}}, \bibinfo {author} {\bibfnamefont {N.}~\bibnamefont {Johnston}},\ and\ \bibinfo {author} {\bibfnamefont {G.}~\bibnamefont {Adesso}},\ }\bibfield  {title} {\bibinfo {title} {Robustness of coherence: an operational and observable measure of quantum coherence},\ }\href@noop {} {\bibfield  {journal} {\bibinfo  {journal} {Physical review letters}\ }\textbf {\bibinfo {volume} {116}},\ \bibinfo {pages} {150502} (\bibinfo {year} {2016})}\BibitemShut {NoStop}%
\bibitem [{\citenamefont {Piani}\ \emph {et~al.}(2016)\citenamefont {Piani}, \citenamefont {Cianciaruso}, \citenamefont {Bromley}, \citenamefont {Napoli}, \citenamefont {Johnston},\ and\ \citenamefont {Adesso}}]{Piani16}%
  \BibitemOpen
  \bibfield  {author} {\bibinfo {author} {\bibfnamefont {M.}~\bibnamefont {Piani}}, \bibinfo {author} {\bibfnamefont {M.}~\bibnamefont {Cianciaruso}}, \bibinfo {author} {\bibfnamefont {T.~R.}\ \bibnamefont {Bromley}}, \bibinfo {author} {\bibfnamefont {C.}~\bibnamefont {Napoli}}, \bibinfo {author} {\bibfnamefont {N.}~\bibnamefont {Johnston}},\ and\ \bibinfo {author} {\bibfnamefont {G.}~\bibnamefont {Adesso}},\ }\bibfield  {title} {\bibinfo {title} {Robustness of asymmetry and coherence of quantum states},\ }\href@noop {} {\bibfield  {journal} {\bibinfo  {journal} {Physical Review A}\ }\textbf {\bibinfo {volume} {93}},\ \bibinfo {pages} {042107} (\bibinfo {year} {2016})}\BibitemShut {NoStop}%
\bibitem [{\citenamefont {Piani}\ and\ \citenamefont {Watrous}(2009)}]{PianiWatrous09}%
  \BibitemOpen
  \bibfield  {author} {\bibinfo {author} {\bibfnamefont {M.}~\bibnamefont {Piani}}\ and\ \bibinfo {author} {\bibfnamefont {J.}~\bibnamefont {Watrous}},\ }\bibfield  {title} {\bibinfo {title} {All entangled states are useful for channel discrimination},\ }\href@noop {} {\bibfield  {journal} {\bibinfo  {journal} {Physical Review Letters}\ }\textbf {\bibinfo {volume} {102}},\ \bibinfo {pages} {250501} (\bibinfo {year} {2009})}\BibitemShut {NoStop}%
\bibitem [{\citenamefont {Takagi}\ and\ \citenamefont {Regula}(2019)}]{Takagi19}%
  \BibitemOpen
  \bibfield  {author} {\bibinfo {author} {\bibfnamefont {R.}~\bibnamefont {Takagi}}\ and\ \bibinfo {author} {\bibfnamefont {B.}~\bibnamefont {Regula}},\ }\bibfield  {title} {\bibinfo {title} {General resource theories in quantum mechanics and beyond: Operational characterization via discrimination tasks},\ }\href@noop {} {\bibfield  {journal} {\bibinfo  {journal} {Physical Review X}\ }\textbf {\bibinfo {volume} {9}},\ \bibinfo {pages} {031053} (\bibinfo {year} {2019})}\BibitemShut {NoStop}%
\bibitem [{\citenamefont {Bernstein}\ and\ \citenamefont {Vazirani}(1997)}]{Bernstein97}%
  \BibitemOpen
  \bibfield  {author} {\bibinfo {author} {\bibfnamefont {E.}~\bibnamefont {Bernstein}}\ and\ \bibinfo {author} {\bibfnamefont {U.}~\bibnamefont {Vazirani}},\ }\bibfield  {title} {\bibinfo {title} {Quantum complexity theory},\ }\href@noop {} {\bibfield  {journal} {\bibinfo  {journal} {SIAM Journal on Computing}\ }\textbf {\bibinfo {volume} {26}},\ \bibinfo {pages} {1411} (\bibinfo {year} {1997})}\BibitemShut {NoStop}%
\bibitem [{\citenamefont {Naseri}\ \emph {et~al.}(2022)\citenamefont {Naseri}, \citenamefont {Kondra}, \citenamefont {Goswami}, \citenamefont {Fellous-Asiani},\ and\ \citenamefont {Streltsov}}]{Naseri22}%
  \BibitemOpen
  \bibfield  {author} {\bibinfo {author} {\bibfnamefont {M.}~\bibnamefont {Naseri}}, \bibinfo {author} {\bibfnamefont {T.~V.}\ \bibnamefont {Kondra}}, \bibinfo {author} {\bibfnamefont {S.}~\bibnamefont {Goswami}}, \bibinfo {author} {\bibfnamefont {M.}~\bibnamefont {Fellous-Asiani}},\ and\ \bibinfo {author} {\bibfnamefont {A.}~\bibnamefont {Streltsov}},\ }\bibfield  {title} {\bibinfo {title} {Entanglement and coherence in the bernstein-vazirani algorithm},\ }\href@noop {} {\bibfield  {journal} {\bibinfo  {journal} {Physical Review A}\ }\textbf {\bibinfo {volume} {106}},\ \bibinfo {pages} {062429} (\bibinfo {year} {2022})}\BibitemShut {NoStop}%
\bibitem [{\citenamefont {Chi}\ \emph {et~al.}(2001)\citenamefont {Chi}, \citenamefont {Kim},\ and\ \citenamefont {Lee}}]{Chi01}%
  \BibitemOpen
  \bibfield  {author} {\bibinfo {author} {\bibfnamefont {D.~P.}\ \bibnamefont {Chi}}, \bibinfo {author} {\bibfnamefont {J.}~\bibnamefont {Kim}},\ and\ \bibinfo {author} {\bibfnamefont {S.}~\bibnamefont {Lee}},\ }\bibfield  {title} {\bibinfo {title} {Initialization-free generalized deutsch-jozsa algorithm},\ }\href@noop {} {\bibfield  {journal} {\bibinfo  {journal} {Journal of Physics A: Mathematical and General}\ }\textbf {\bibinfo {volume} {34}},\ \bibinfo {pages} {5251} (\bibinfo {year} {2001})}\BibitemShut {NoStop}%
\bibitem [{\citenamefont {Watrous}(2018)}]{watrous2018theory}%
  \BibitemOpen
  \bibfield  {author} {\bibinfo {author} {\bibfnamefont {J.}~\bibnamefont {Watrous}},\ }\href@noop {} {\emph {\bibinfo {title} {The theory of quantum information}}}\ (\bibinfo  {publisher} {Cambridge university press},\ \bibinfo {year} {2018})\BibitemShut {NoStop}%
\end{thebibliography}%

\end{document}